\documentclass[runningheads]{llncs}

\usepackage{xspace}

\usepackage{graphicx}

\usepackage{color}
\usepackage{amsmath}

\usepackage[autolanguage,np]{numprint}

\usepackage{booktabs} 

\usepackage{algorithm}
\usepackage[noend]{algpseudocode}

\usepackage{hyperref}
\usepackage[utf8]{inputenc}

\usepackage{siunitx}
\usepackage[autolanguage,np]{numprint}
\usepackage[toc,page]{appendix}

\usepackage{flushend}
\usepackage{multirow}
\usepackage{microtype} 
\usepackage{enumitem}

\algloopdefx[loop]{Iterate}
[1][]{{\bf iterate} }

\newcommand{\cpm}{\textsc{cpm}\xspace}
\newcommand{\cpmz}{\textsc{cpmz}\xspace}

\newcommand{\setz}{\texttt{Setz}\xspace}
\newcommand{\dict}{\texttt{Dict}\xspace}
\newcommand{\find}{\texttt{Find}\xspace}
\newcommand{\union}{\texttt{Union}\xspace}
\newcommand{\makeset}{\texttt{MakeSet}\xspace}
\newcommand{\add}{\texttt{Add}\xspace}
\newcommand{\uf}{\texttt{UF}\xspace}

\usepackage{tikz}
\usetikzlibrary{patterns}
\usetikzlibrary{decorations.pathreplacing}
\usepackage{bold-extra}
\usepackage{amssymb}

\colorlet{sk}{blue}
\colorlet{cm}{purple}
\usepackage[normalem]{ulem}

\begin{document}
\title{Clique percolation method:
  memory efficient  almost exact communities}

\author{Alexis Baudin\inst{1} \and
  Maximilien Danisch\inst{1} \and
  Sergey Kirgizov\inst{2} \and
  Clémence Magnien\inst{1} \and
  Marwan Ghanem\inst{1}
}
\authorrunning{A. Baudin {\em et al.}}

\institute{Sorbonne Université, CNRS, LIP6, F-75005 Paris, France\\
  \email{firstname.lastname@lip6.fr}\\
  \and
  LIB, Université de Bourgogne Franche-Comté\\
  B.P. 47 870, 21078 Dijon Cedex France\\
  \email{sergey.kirgizov@u-bourgogne.fr}
}

\pagenumbering{gobble}
\maketitle

\begin{abstract} 
  Automatic detection of relevant groups of nodes in large real-world
  graphs, i.e. community detection, has applications in many fields and
  has received a lot of attention in the last twenty years. The most
  popular method designed to find overlapping communities (where a node
  can belong to several communities) is perhaps the clique percolation
  method (\cpm). This method formalizes the notion of community as a
  maximal union of $k$-cliques that can be reached from each other
  through a series of adjacent $k$-cliques, where two cliques are
  adjacent if and only if they overlap on $k-1$ nodes.  Despite much effort \cpm has
  not been scalable to large graphs for medium values of $k$.

  Recent work has shown that it is possible to efficiently list all
  $k$-cliques in very large real-world graphs for medium values of
  $k$. We build on top of this work and scale up \cpm. In cases where this
  first algorithm faces memory limitations, we propose
  another algorithm, \cpmz, that provides a solution close to the exact one,
  using more time but less memory.


\end{abstract} 

\section{Introduction}
\label{sec:intro}

The problem of detecting communities in real networks has received a
lot of attention in recent years. Many definitions of communities have
been proposed, corresponding to different requirements on the type of
communities that are to be detected and/or on some properties of the
studied graph: some definitions depend on global or local properties
of the graph, nodes can belong to several communities
or to a single one, links may have weights, communities may have a
hierarchical structure, {\em etc}.
In practice, algorithms also have to be designed to extract the
communities from large graphs.  Many definitions of communities proposed with
their corresponding algorithm already exist~\cite{Fortunato2012}. Most
real networks are characterized by communities that may overlap,
i.e. a node may belong to several communities.  
In the context of social networks for instance, each person belongs to
several communities such as colleagues, family, leisure activities,
etc.

One of the most popular methods designed to find overlapping
communities is the clique percolation method (\cpm) which
produces {\em $k$-clique communities}~\cite{palla2005uncovering}.

\begin{definition}[$k$-clique]
  A $k$-clique $c_k$ is a fully connected set of $k$ nodes,
  i.e. every pair of its vertices is connected by a link in the graph.
\end{definition}

For example, in Figure~\ref{gr1}, the set $\{1,3,4,6 \}$ is a $4$-clique.

\begin{definition}[$k$-cliques adjacency]
  Two $k$-cliques are said to be adjacent if and
  only if they share $k-1$ nodes.
\end{definition}

For example, in Figure~\ref{gr1}, the two $4$-cliques $\{1,3,4,6 \}$ and  \(\{1,3,6,9 \} \) are adjacent.

\begin{definition}[\cpm community]
  A $k$-clique community (or \cpm community) is the set of the
  vertices belonging to a maximal set of $k$-cliques that can be
  reached from each other through a series of adjacent $k$-cliques.
\end{definition}

Though the corresponance between the obtained communities and real-world ones are hard
to charactarize in a general manner,
the advantages of this definition are well
known~\cite{gregori2013parallel}: it is formally well defined, totally
deterministic, does not use any heuristics or function optimizations
that are hard to interpret, allows communities to overlap, and each
community is defined  locally\footnote{Notice that if a node does not belong
  to at least one $k$-clique, it doesn't belong to any community.}.

Despite much effort \cpm has not been scalable to large graphs for medium values of $k$,
i.e, values between 5 and 10.
We therefore seek in this work to extend the computation of \cpm to larger graphs.
A bottleneck for most previous contributions is the memory required.
Indeed, exact methods need to store in memory either all $k$-cliques, or all maximal
cliques (cliques which are not included in any other clique), which is prohibitive in many cases.

Our contribution is twofold:
\begin{enumerate}[topsep=0ex,partopsep=0ex]
\item we improve on the state of the art concerning the computation of \cpm{} communities,
  by leveraging an existing algorithm able to list $k$-cliques in a very efficient manner;
\item in cases where this first algorithm faces memory limitations,
  we propose another algorithm that provides a solution close to the exact one,
  using more time but less memory.
\end{enumerate}

We will show that these algorithms allow to compute exact solutions in cases where
this was not possible before, and to compute a close result of good quality in cases
where the graph is so large that our exact method does not work due to memory limitations.

The rest of the paper is organized as follows. In Section
\ref{sec:related}, we present the related work. In Section
\ref{sec:algorithm}, we present our exact and relaxed algorithms. We discuss
time and memory requirements in Section \ref{sec:analysis}. We then
evaluate the performance of our exact algorithm against the state of the art
in Section~\ref{sec:experiments}, and we compare the results and performances of our
exact and relaxed algorithms.
We conclude in Section~\ref{conc}, and present some perspectives for future work.


\section{Related Work}
\label{sec:related}

There are many algorithms for computing overlapping communities as
shown in a dedicated survey~\cite{xie2013overlapping}.  The focus of
our paper is on the computation of the $k$-clique communities.
Existing algorithms to compute $k$-clique communities in a graph can
be split in two categories:
\begin{enumerate}[topsep=0ex,partopsep=0ex]
\item[(1)]
  algorithms that compute all maximal cliques of size $k$ or
  more and use them to compute all $k$-clique communities.
  Indeed, two maximal cliques that overlap on $k-1$ nodes or
  more belong to the same $k$-clique community. Most state-of-the-art
  approaches~\cite{gregori2013parallel,palla2005uncovering,reid2012percolation}
  belong to this category;
\item[(2)] Kumpula {\em et al.}~\cite{kumpula2008sequential}
  compute all $k$-cliques and then compute $k$-clique communities from them strictly
  following the definitions of a $k$-clique community, i.e. detect which $k$-cliques
  are adjacent.
\end{enumerate}
Algorithms of the first category differ in the method used to find which maximal cliques
are adjacent. However, the first step which consists in computing all maximal cliques is
always the same and is done sequentially.
While this problem is NP-hard, there exist algorithms scalable to relatively large sparse real-world graphs,
based on the Bron-Kerbosch algorithm~\cite{bron1973algorithm,eppstein2010listing,eppstein2011listing}. 

Any large clique, with more than $k$ vertices, will be included in
a single $k$-clique community, and there is no need to list all
$k$-cliques of this large clique.  This is the main reason why there
are more methods following the approach of listing maximal cliques,
category~(1), rather than listing $k$-cliques, category~(2).
However, it has been found that most real-world graphs actually do
not contain very large cliques and that listing $k$-cliques for small and medium
values of $k$ is a scalable problem in practice~\cite{DanischBS18,Li2020Ordering},
in many cases it is more tractable
that listing all maximal cliques. This makes algorithms in
the category~(2) more interesting for practical scenarios.

The algorithm of~\cite{kumpula2008sequential} proposes a method to
list all $k$-cliques then merges the found $k$-cliques into $k$-clique
communities using a {\em Union-Find}~\cite{gal}, a very efficient data
structure~\cite{fre,tar} which we describe briefly in Section~\ref{sec:uf}.
In the context of~\cite{kumpula2008sequential} the
Union-Find contains all $(k-1)$-cliques (as elements) and each
$k$-clique $c_k$ triggers the union of the subsets that contain
at least one $(k-1)$-clique of $c_k$.

Our first contribution builds on the same idea. We first propose to
use an efficient algorithm for listing $k$-cliques~\cite{DanischBS18},
which improves the overall performance.  Going further, in order to
provide an approximation of the community structure for graphs for
which it is not possible to obtain the exact result due to memory
limits, we propose to perform union of sets of $z$-cliques, $2 \le z < k-1$,
instead of $(k-1)$-cliques. This construction is discussed in details
in the next section.


\section{Algorithm}
\label{sec:algorithm}

A graph $G=(V,E)$ consists of its vertex set $V$ and its edge set $E$.
In the following $c_k$ will always denote a
$k$-clique, from the context it will be clear which one exactly.

\subsection{Union-Find structure}
\label{sec:uf}

The algorithms we will present rely on 
the Union-Find structure, also known in the literature as a {\em
  disjoint-set data structure}.
It stores a collection of disjoint sets,
allowing very efficient union operations between them. The structure
is a forest, whose trees correspond to disjoint subsets, and nodes
correspond to the elements.
The operations on the nodes are the following:
\begin{itemize}[topsep=0ex,partopsep=0ex]
\item $\find(p)$: returns the root of the tree containing a Union-Find node $p$.
\item $\union(r_1, ..., r_l)$:  performs the union of trees
  represented by their roots \texttt{$r_i$} by making one root the parent of all others;
\item $\makeset()$: creates a new tree with one node $p$, corresponding to a new empty set,
  and returns $p$.
\end{itemize}

\subsection{Exact \cpm algorithm}

First we build on the idea introduced in~\cite{kumpula2008sequential}.
A \cpm{} community is represented as the set of all the $(k-1)$-cliques it contains.
These communities are represented by a Union-Find structure whose nodes are $(k-1)$-cliques.
The algorithm then iterates over all $k$-cliques and tests
if the current $k$-clique belongs to a community, by testing whether it has a $(k-1)$-clique
in common with it.

\begin{algorithm}[!htbp]
  \caption{Exact \cpm algorithm}
  \label{algo:cpm}
  \begin{algorithmic}[1]
    \State ${\tt UF} \gets$ Union-Find data structure
    \State $\dict \gets$ Empty Dictionary
    \For {each $k$-clique $c_k \in G$}
    \State $S \gets \emptyset$
    \Comment communities of  $c_k$ to merge
    \For {each $(k-1)$-clique $c_{k-1} \subset c_k$}
    \If {$c_{k-1} \in \dict.\texttt{keys}()$}
    \State $p \gets \uf.\find(\dict[c_{k-1}])$ \label{cpm:Find}
    \Else
    
    \State $p \gets {\tt UF.MakeSet}()$
    \State $\dict[c_{k-1}] \gets p$
    \EndIf
    \State $S \gets S \cup \{p\}$
    \EndFor
    \State $\uf.\union(S)$
    \EndFor
  \end{algorithmic}
\end{algorithm}

Algorithm~\ref{algo:cpm} considers all $k$-cliques one by one. For every $k$-clique
it iterates over its $(k-1)$-cliques $c_{k-1}^1, c_{k-1}^2, \ldots c_{k-1}^k$.
For every $c_{k-1}^i, i \in [1,k]$, it identifies the set $p_i$ to which it belongs in the
Union-Find. Then, it performs the union of all sets $p_i$.
Several algorithms exist for efficiently listing $k$-cliques~\cite{Li2020Ordering}.
We substitute one the best~\cite{DanischBS18} to the one proposed by the authors
of~\cite{kumpula2008sequential}.

As the number of $(k-1)$-cliques can be very large, this approach is problematic as
in some cases it is not possible to store them all in memory.
This leads us to a new algorithm which requires less memory but in rare cases incorrectly
merges some \cpm communities together.

\subsection{Memory efficient \cpm approximation}

For relatively small values of $k$, there are far fewer $z$-cliques than $(k-1)$-cliques
in real-world graphs. To get an intuition for this, consider the case of a large clique
of size $c$. It contains $c \choose z$ $z$-cliques and this number increases with $z$
for $z<c/2$. Therefore storing all $z$-cliques is feasible in cases where it is not possible
to store all the $(k-1)$-cliques.
We use this idea to propose an algorithm computing relaxed communities.

\begin{definition}[\cpmz community]
  An agglomerated $k$-clique community (or \cpmz community) is the
  union of one or more \cpm{} communities.
\end{definition}

Our memory efficient method, called {\em  \cpmz algorithm}, given a graph $G$, the size
of $k$-cliques and an integer $z \in [2, k-1)$, returns a set of agglomerated
$k$-clique communities, such that each \cpm{} community is included
in one and only one \cpmz community (see Theorem~\ref{thm}).

In the \cpm algorithm, a community is represented as a set of $(k-1)$-cliques, and
communities correspond to {\em disjoint} sets of $(k-1)$-cliques.
In the following, a community is represented as a set of $z$-cliques, and \cpmz
communities are represented as \emph{non-disjoint} sets of $z$-cliques.

The main idea of our \cpmz algorithm is to identify each $(k-1)$-clique to the set of
its containing $z$-cliques. The algorithm is very similar to \cpm.
For each $k$-clique, all $(k-1)$-cliques are considered.
Since we consider that a $(k-1)$-clique is represented by the set of its $z$-cliques,
the community of a $(k-1)$-clique is the one that contains all its $z$-cliques.

The \cpmz algorithm uses two principal data structures. \texttt{UF} is an Union-Find data
structure, whose nodes are identifiers of $z$-clique sets.
We will call these {\em Union-Find nodes}. The operations defined on this structure are
presented in Section~\ref{sec:uf}. Each $z$-clique can belong to several Union-Find nodes.
This is recorded in the  \texttt{Setz} dictionary, which associates to each $z$-clique the
set of Union-Find nodes to which it belongs.
See Figure~\ref{fig:ex-ouf} for an example.

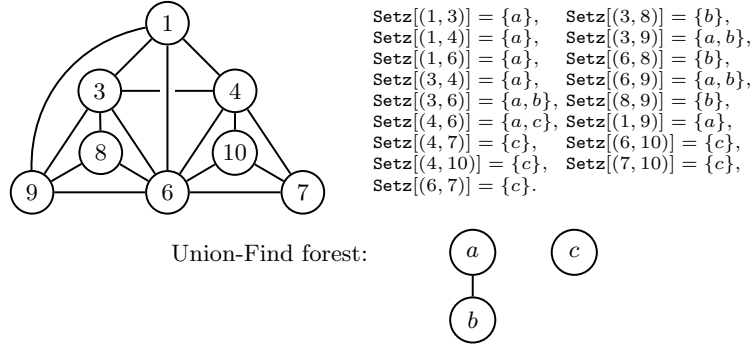
\begin{figure}
  \centering

  \tikzset{every picture/.style={line width=0.75pt}} 

  \begin{tikzpicture}[scale=0.9]
    \tikzstyle{nod1} = [circle, draw, fill=black, scale=1.0,
    text = white];
    \tikzstyle{nod2} = [circle, draw, fill=white, scale=1.0];
    \tikzstyle{gd} = [black,-];

    \useasboundingbox (-2.5,0.5) rectangle (2.5,-3);
    
    \node[nod2] (1) at  (0,0){1};
    \node[nod2] (3) at  (-1,-1){3};
    \node[nod2] (4) at  (1,-1){4};
    \node[nod2] (a6) at  (0,-2.5){6};
    \node[nod2] (a7) at  (2,-2.5){7};
    \node[nod2, inner sep=2.1pt] (10) at (1.38,-1.53,){10};
    \node[nod2] (a8) at (-0.6,-1.53,){8};
    \node[nod2] (a9) at (-2,-2.5){9};
    
    \path[gd] (1) edge  (3);
    \path[gd] (1) edge  (4);
    
    \path[gd] (3) edge  (4);
    \path[gd] (3) edge  (a6);

    \path[gd, line width=2mm, white] (1) edge  (a6);    
    \path[gd] (1) edge  (a6);

    \path[gd] (4) edge  (a6);
    \path[gd] (4) edge  (a7);
    
    \path[gd] (a6) edge  (a7);

    \path[gd] (1) edge
    [bend right=40] (a9);         
    
    \path[gd] (3) edge  (a9);
    \path[gd] (a6) edge  (a9);
    
    \path[gd] (4) edge  (10);
    \path[gd] (a6) edge  (10);
    \path[gd] (a7) edge  (10);

    \path[gd] (a8) edge  (a6);
    \path[gd] (a8) edge  (a9);
    
    \path[gd] (a8) edge  (3);

  \end{tikzpicture}\quad
  \raisebox{1.2\height}{
    \scriptsize
    $\begin{array}{ll}
       \setz[(1,3)] = \{a\},   &  \setz[(3,8)] = \{b\},  \\
       \setz[(1,4)] = \{a\},   &  \setz[(3,9)] = \{a,b\},  \\
       \setz[(1,6)] = \{a\},   &  \setz[(6,8)] = \{b\},  \\
       \setz[(3,4)] = \{a\},   &  \setz[(6,9)] = \{a,b\},  \\
       \setz[(3,6)] = \{a,b\}, &  \setz[(8,9)] = \{b\},  \\
       \setz[(4,6)] = \{a,c\}, &  \setz[(1,9)] = \{a\},  \\
       \setz[(4,7)] = \{c\},   &  \setz[(6,10)] = \{c\},  \\
       \setz[(4,10)] = \{c\},  &  \setz[(7,10)] = \{c\},  \\
       \setz[(6,7)] = \{c\}.   &  \\
     \end{array}$
   }

   \begin{tikzpicture}[scale=0.9]
     \tikzstyle{nod1} = [circle, draw, fill=black, scale=1.0,
     text = white];
     \tikzstyle{nod2} = [circle, draw, fill=white, minimum width=6mm];
     \tikzstyle{gd} = [black,-];

     \node (a) at  (-3,0){Union-Find forest:};
     \node[nod2] (a) at  (0,0){$a$};
     \node[nod2] (b) at  (0,-1){$b$};
     \node[nod2] (c) at  (1.5,0){$c$};

     \path[gd] (a) edge  (b);
   \end{tikzpicture}
   \caption{ Example of a graph and the corresponding data structures of
     the \cpmz{} algorithm.  In this example, we have $k=4, z = 2$.
     There are 17 $z$-cliques belonging to 4 $k$-cliques, namely
     \{1,3,4,6\}, \{3,6,8,9\}, \{4,7,6,10\} and \{1,3,6,9\}.  The nodes
     of the Union-Find structure are represented using letters $a,b$
     and $c$.  Each $z$-clique is associated to one or more Union-Find
     nodes, as shown in the \setz information on the top-right.  The
     Union-Find structure represents two sets because there are two
     different root nodes: $a$ and $c$.  The first set contains all
     2-cliques associated with $a$ or $b$, the second contains
     2-cliques associated with $c$.}
   \label{fig:ex-ouf}
 \end{figure}

More formally, \texttt{Setz} is a dictionary with $z$-cliques as keys.
For a $z$-clique $c_z$:
\begin{itemize}
\item \texttt{Setz[$c_z$]} is a set of Union-Find nodes;
\item \texttt{Setz[$c_z$].add($q$)} adds the Union-Find node $q$
  to the set of Union-Find nodes of $c_z$.
  It can also be seen as the action of adding $c_z$
  to the set identified by $q$.
\end{itemize}

At the end of the algorithm, every tree corresponds to a \cpmz community
represented as a union of sets containing $z$-cliques. 

Note that during the execution of the algorithm, the same $z$-clique $c_z$ can belong to several
Union-Find nodes of the same Union-Find set, which creates redundancies in \setz[$c_z$].
This is the case for instance in the example of Figure~\ref{fig:ex-ouf} in which \texttt{Setz[(3,6)]} contains
both $a$ and $b$ which belong to the same Union-Find set.
This situation can occur if a $z$-clique belongs to two different Union-Find sets which are merged later.

In our \cpmz{} algorithm (presented below) we eliminate these redundancies when
we detect them (see Line~\ref{cpmz:reduce-setz}).

\begin{algorithm}[!htbp]
  \begin{algorithmic}[1]
    \State  ${\tt UF} \gets$ Empty Union-Find data structure
    \State ${\tt Setz} \gets$ Empty Dictionary
    \For {each $k$-clique $c_k \in G$ } \label{cpmz:for:kcliques}
    \State $S\gets \emptyset$ \label{cpmz:initS}
    \Comment{Sets of $z$-cliques to merge}
    \For {each $(k-1)$-clique $c_{k-1} \subset c_k$ \label{cpmz:for:k-1cliques}}
    \State $P \gets \emptyset$
    \For {each $z$-clique $c_z \subset c_{k-1}$}
    \State $\setz[c_z] \gets \{{\tt UF.Find}(p) \ for \ p \in {\tt Setz}[c_z]\}$
    \label{cpmz:reduce-setz}
    \If {$P == \emptyset$}
    \State $P \gets \setz[c_z]$
    \Else
    \State $P \gets P \cap \setz[c_z]$
    \label{cpmz:inter}
    \EndIf
    \EndFor
    \State $S \gets S\cup P$
    \label{cpmz:for:k-1cliques:end}
    \EndFor
    \State $q \gets NULL$
    \Comment{Identifier of the resulting set of $z$-cliques}
    \If {$S == \emptyset$}
    \State  $q \gets {\tt UF.MakeSet()}$  \label{cpmz:makeset}
    \Else
    \State  $q\gets {\tt UF.Union}(S)$  \label{cpmz:unionS}
    \EndIf
    \For {each $z$-clique $c_z \subset c_{k}$}
    \State \texttt{Setz[$c_z$].add($q$)}  \label{cpmz:add}
    \EndFor
    \EndFor
  \end{algorithmic}
  \caption{\cpmz pseudocode}
  \label{algo:cpmz}
\end{algorithm}

Algorithm~\ref{algo:cpmz} is the pseudo-code of \cpmz. The for loop on
Line~\ref{cpmz:for:kcliques} iterates over each $k$-clique $c_k$ of a graph $G$.
As in the \cpm{} algorithm, the idea is to identify the communities of
each $(k-1)$-clique of $c_k$ and perform their union.
As explained above, the communities of a $(k-1)$-clique are the ones that contain all
its $z$-cliques, which is why we compute their intersection in the set $P$ in
Line~\ref{cpmz:inter}.
The set $P$ is computed for all $(k-1)$-cliques in Lines~\ref{cpmz:initS}-\ref{cpmz:for:k-1cliques:end}
and their union is computed in set $S$.
Then all the sets in $S$ are
merged in Line~\ref{cpmz:unionS}.

It may turn out that
$S$ is empty after the loop of Line~\ref{cpmz:for:k-1cliques}.  This
corresponds to the case where none of the $(k-1)$-cliques of $c_k$
were observed before:
if a $(k-1)$-clique $c_{k-1}$ has not yet been seen in the algorithm,
its $z$-cliques may not belong to a common Union-Find set,
and therefore $P$ computed at Line~\ref{cpmz:inter} is empty.
If this happens for all $(k-1)$-cliques  of $c_k$ S is empty
and a new set (Union-Find node) is created on Line~\ref{cpmz:makeset}.
The identifier of the resulting (new or merged) set is added to the set
of Union-Find nodes for every $z$-clique of the current $k$-clique
(Line~\ref{cpmz:add}).

In some rare cases, Line~\ref{cpmz:inter} will consider that a $(k-1)$-clique
belongs to a Union-Find set while this is not true in the \cpm exact case.
Figure~\ref{gr1} gives an example. In that case this causes an incorrect $k$-clique
adjacency detection and results in an incorrect merge of two or more $k$-clique communities.

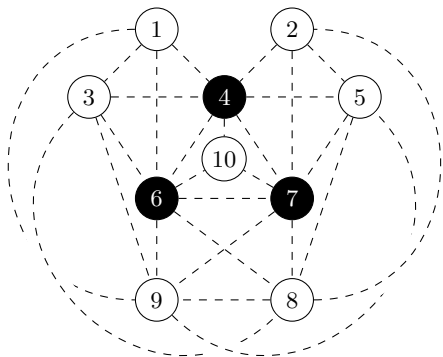
\begin{figure}
  \centering

  \begin{tikzpicture}[scale=0.9, rotate=0]
    \tikzstyle{nod1} = [circle, draw, fill=black, scale=1.0,
    text = white];
    \tikzstyle{nod2} = [circle, draw, fill=white, scale=1.0];
    \tikzstyle{gd} = [black,dashed,-];

    \useasboundingbox (-2.5,0.5) rectangle (4.5,-4.8); 

    \node[nod2] (1) at  (0,0){1};
    \node[nod2] (2) at  (2,0){2};
    
    \node[nod2] (3) at  (-1,-1){3};
    \node[nod1] (4) at  (1,-1){4};
    \node[nod2] (a5) at  (3,-1){5};
    \node[nod1] (a6) at  (0,-2.5){6};
    \node[nod1] (a7) at  (2,-2.5){7};
    \node[nod2] (a8) at (2,-4){8};
    \node[nod2] (a9) at (0,-4){9};
    \node[nod2, inner sep=2.1pt] (10) at (1.38,-1.53,){10};
    
    \path[gd] (1) edge  (3);
    \path[gd] (1) edge  (4);
    \path[gd] (1) edge  (a6);
    
    \path[gd] (2) edge  (a5);
    \path[gd] (2) edge  (4);
    \path[gd] (2) edge  (a7);
    \path[gd] (3) edge  (4);
    \path[gd] (3) edge  (a6);

    \path[gd] (4) edge  (a6);
    \path[gd] (4) edge  (a7);
    \path[gd] (4) edge  (a5);
    
    \path[gd] (a5) edge  (a7);

    \path[gd] (a6) edge  (a7);
    \path[gd] (a6) edge  (a9);
    \path[gd] (a6) edge  (a8);
    
    \path[gd] (a7) edge  (a9);
    \path[gd] (a7) edge  (a8);

    \path[gd] (3) edge (a9);
    \path[gd] (a5) edge (a8);

    \path[gd] (a8) edge  (a9);
    \path[gd] (4) edge  (10);
    \path[gd] (a6) edge  (10);
    \path[gd] (a7) edge  (10);

    \path[gd] (1) edge
    [bend right=90, looseness=1.6] (a9);
    
    \path[gd] (a9) edge 
    [bend right=90, looseness=1.6] (a5);

    \path[white,-] (3) edge 
    [bend right=90, looseness=1.6, line width=3mm
    ] (a8);
    \path[gd] (3) edge 
    [bend right=90, looseness=1.6] (a8);

    \path[white,-] (a9) edge 
    [bend right=90, looseness=1.6, line width=3mm
    ] (a5);
    \path[gd] (a9) edge 
    [bend right=90, looseness=1.6] (a5);

    \path[white,-] (a8) edge
    [bend right=90, looseness=1.6, line width=3mm
    ] (2);
    \path[gd] (a8) edge 
    [bend right=90, looseness=1.6] (2);

    \end{tikzpicture}
  \caption{ In this example $k=4$ and $z=2$. A $k$-clique percolation
    community is formed by the nodes of the following $k$-cliques:
    $\{1,3,4,6\}$, $\{1,3,6,9\}$, $\{3,6,8,9\}$, $\{6,7,8,9\}$,
    $\{5,7,8,9\}$, $\{2,5,7,8\}$, $\{2,4,5,7\}$.  The middle
    $(k-1)$-clique with nodes $\{4,6,7\}$ is formed by the $z$-cliques
    (edges) of other $k$-cliques, whereas it is not itself a part of any
    $k$-clique given above. When a new $k$-clique $\{4,6,7,10\}$ is
    observed, the \cpmz algorithm will produce one community
    $\{1,2,3,4,5,6,7,8,9,10\}$, but the exact \cpm algorithm gives two
    communities, namely $\{1,2,3,4,5,6,7,8,9\}$ and $\{4,6,7,10\}$.
  }
  \label{gr1}
\end{figure}

\begin{theorem}[\cpmz validity]
  The \cpmz algorithm returns a set of agglomerated $k$-cliques
  communities, such that each \cpm community is included in one and
  only one agglomerated community.
  \label{thm}
\end{theorem}

\begin{proof}
  If two $k$-cliques $c_k^1$ and $c_k^2$ are adjacent, this means they share a $(k-1)$-clique $c_{k-1}$.
  This will be correctly detected by Lines~\ref{cpmz:for:k-1cliques}-\ref{cpmz:inter} of Algorithm~\ref{algo:cpmz}:
  after the iteration on $c_k^1$ in the main loop,
  all $z$-cliques of $c_{k-1}$ will belong to a common Union-Find set.
  The root of this set will belong to $S$ during the iteration on $c_k^2$,
  ensuring that the Union-Find sets of both $k$-cliques will be merged.
  In other words, \cpm communities are never   split by \cpmz Algorithm and
  each \cpm{} community belongs to a single agglomerated community.
  Conversely, an agglomerated community may contain more than one \cpm{} community.
\end{proof}


\section{Analysis}
\label{sec:analysis}

We denote the number of $k$-cliques in graph $G$ by $n_k$.  For each
$k$-clique, the \cpm{} algorithm performs a \find and a \union operation
for each of its $(k-1)$-clique (see Algorithm~\ref{algo:cpm}). It is well
known (see for example~\cite{tar}) that Union-Find data structure performs
operations in $O(\alpha(n))$ amortized time, where $\alpha$ is the inverse
Ackermann function, and $n$ is the number of elements.  It grows extremely slowly.
Therefore, the number of operations is proportional in practice to $k\cdot n_k$.
The space complexity of this algorithm is dominated by the tree on the
$(k-1)$-cliques and the corresponding cost is
proportional to $(k-1)\cdot n_{k-1}$.

\medskip
\cpmz is a tradeoff between memory and time.
Indeed, we will see that the \cpmz{} has a higher running time than the exact
\cpm{} algorithm, but requires less memory.

For the \cpmz{} algorithm as well, the time required by the \union and
\find operations can still be considered as constant, as is the time
required for the \makeset and \add operations. The total number of operations
is then dominated by the number of \find operations of Line~\ref{cpmz:reduce-setz}.
This line runs on each distinct triplet $(c_k,c_{k-1},c_z)$, where
$c_{k-1} \subset c_k$ and $c_z \subset c_{k-1}$, and there are
$k-1 \choose z$ $\cdot k \cdot n_k $ such triplets. Each $z$-clique $c_z$ of Line~\ref{cpmz:reduce-setz} belongs to a number of Union-Find nodes
$|\setz[c_z]|$ that depends on the considered clique but also varies during
the execution of the algorithm:  it can either increase as new Union-Find node
are added to $\setz[c_z]$ (on Line~\ref{cpmz:add}) or decrease
(on Line~\ref{cpmz:reduce-setz}) if some Union-Find trees are merged
(on Line~\ref{cpmz:unionS}).
This number $|\setz[c_z]|$ is bounded by the number of $k$-cliques each
$z$-clique belongs to, which can theoretically be quite high.  However, we computed
in practice the average number of \find operations performed for each $z$-clique, and
we will see in Section~\ref{sec:cpmz-precision} that it is very often equal to 1 or 2
and never exceeds 6 in our experiments.
The main difference in the running time with respect to the exact \cpm{} algorithm is,
therefore, the extra $k-1 \choose z$ factor.

Concerning the space requirements, the \cpmz{} algorithm needs to
store all $z$-cliques, which takes a space proportional to $z\cdot
n_z$.  Each $z$-clique $c_z$ then belongs to a number of  Union-Find
nodes $|\setz[c_z]|$ that varies during the algorithm execution.
Even if the average of this number is low, we are interested now in
the maximum space taken at any time of the execution.  Finally, the
number of nodes in the Union-Find structure is equal to the number
of \makeset operations that have been performed during the execution.
In theory this number can also be quite high.  However, we will see in
Section~\ref{sec:experiments} that in practice the memory requirements
of the \cpmz{} algorithm are much lower than those of the \cpm{}
algorithm.

Finally, notice that both our algorithms result in the Union-Find
structure whose nodes represent
sets of cliques.
This structure encodes the
communities.  In order to obtain the actual node list of each
community, post-processing is needed.
  It consists in iterating over all cliques in the Union-Find structure.
  Then for each clique one must find its root node in the Union-Find and add
  the clique nodes to the corresponding set.
We do not take into account this post-processing in the  reported
running time and memory usage in the next section.


\section{Experimental evaluation}
\label{sec:experiments}

\paragraph{Machine}

We carried out our experiments on a Linux machine DELL PowerEdge R440, equipped with 2 processors
Intel Xeon Silver 4216 with 32 cores each, and with 384Gb of RAM. 

\paragraph{Datasets}

We consider several real-world graphs that we obtained from~\cite{snapnets}.
Their characteristics are presented in Table~\ref{tab:graphs}.
We distinguish between three categories of graphs according to their number of $k$-cliques.
For graphs with small core values all algorithms are able to run;
for graphs with medium core values, the state-of-the-art algorithms take too long to complete
(with the exception of DBLP discussed below) while our \cpm{} algorithm
obtains results for small values of $k$;
for graphs with large core values even our \cpm{} algorithm runs out of memory or time except
for very small values of $k$, but we will show that we are able to obtain relaxed results of
high quality with our \cpmz{} algorithm.

\begin{table}
  \caption{Our dataset of real-world graphs, ordered by core value $c$.
    $k_{min}$ and $k_{max}$ represent the minimum and maximum $k$ on which we could run our \cpm
    algorithm. $n_k$ is the number of $k$-cliques of the graph.}
  \label{tab:graphs}
  \centering
  \begin{tabular}{|c|c|c|c|c|c|c|}
    \hline
    network
    & $n$
    & $m$
    & $c$
    & $k_{min}-k_{max}$
    & $n_{k_{min}}$
    & $n_{k_{max}}$
    \\ \hline \hline
    \begin{tabular}[c]{@{}c@{}}soc-pokec\\ loc-gowalla\\ Youtube\\ zhishi-baidu\end{tabular}
    & \begin{tabular}[c]{@{}c@{}}\np{1632803}\\ \np{196591}\\ \np{1134890}\\ \np{2140198}\end{tabular}
    & \begin{tabular}[c]{@{}c@{}}\np{22031964}\\ \np{950327}\\ \np{2987624}\\ \np{17014946}\end{tabular}
    & \begin{tabular}[c]{@{}c@{}}47\\ 51\\ 51\\ 78\end{tabular}
    & \begin{tabular}[c]{@{}c@{}}3 - 15\\ 3 - 15\\ 3 - 15\\ 3 - 15\end{tabular}
    & \begin{tabular}[c]{@{}c@{}}\np{32557458}\\ \np{2273138}\\ \np{3056386}\\ \np{25207196}\end{tabular}
    & \begin{tabular}[c]{@{}c@{}}\np{353958854}\\ \np{201454150}\\ \np{1068}\\ \np{1080702188}\end{tabular}
    \\ \hline \hline
    \begin{tabular}[c]{@{}c@{}}as-skitter\\ DBLP\\ WikiTalk\end{tabular}
    & \begin{tabular}[c]{@{}c@{}}\np{1696415}\\ \np{425957}\\ \np{2394385}\end{tabular}
    & \begin{tabular}[c]{@{}c@{}}\np{11095298}\\ \np{1049866}\\ \np{4659565}\end{tabular}
    & \begin{tabular}[c]{@{}c@{}}111\\ 113\\ 131\end{tabular}
    & \begin{tabular}[c]{@{}c@{}}3 - 6\\ 3 - 7\\ 3 - 7\end{tabular}
    & \begin{tabular}[c]{@{}c@{}}\np{28769868}\\ \np{2224385}\\ \np{9203519}\end{tabular}
    & \begin{tabular}[c]{@{}c@{}}\np{9759000981}\\ \np{60913718813}\\ \np{5490986046}\end{tabular}
    \\ \hline \hline
    \begin{tabular}[c]{@{}c@{}}Orkut\\ Friendster\\ LiveJournal\end{tabular}
    & \begin{tabular}[c]{@{}c@{}}\np{3072627}\\ \np{124836180}\\ \np{4036538}\end{tabular}
    & \begin{tabular}[c]{@{}c@{}}\np{117185083}\\ \np{1806067135}\\ \np{34681189}\end{tabular}
    & \begin{tabular}[c]{@{}c@{}}253\\ 304\\ 360\end{tabular}
    & \begin{tabular}[c]{@{}c@{}}3 - 5\\ 3 - 4\\ 3 - 4\end{tabular}
    & \begin{tabular}[c]{@{}c@{}}\np{627584181}\\ \np{4173724124}\\ \np{177820130}\end{tabular}
    & \begin{tabular}[c]{@{}c@{}}\np{15766607860}\\ \np{8963503236}\\ \np{5216918441}\end{tabular}
    \\ \hline
  \end{tabular}
\end{table}

\paragraph{Implementation}
We implemented our \cpm and \cpmz algorithm in C. The implementation is available on the following gitlab
repository: \url{https://gitlab.lip6.fr/baudin/cpm-cpmz}. 
For the competitors, we used the publicly
available implementations of their
algorithms~\cite{gregori2013parallel,kumpula2008sequential,palla2005uncovering,reid2012percolation}.

\paragraph{Computing domain}

For each graph and each algorithm, we performed the computations of \cpm for all values of $k$
from 3 to 15, unless we were not able to finish the computation for one of the following
reasons:
\begin{itemize}
\item the memory exceeded 390 Gb of RAM,
\item or the computation time exceeded 72 hours.
\end{itemize} 

We ran the
\cpmz algorithm
for $z=2$ and $z=3$. It could be computed on all the cases for which \cpm works, except for $z=3$ in graphs
zhishi-baidu with $k=15$ and DBLP with $k=7$.

The interesting point is that we manage to have results with \cpmz in cases where the computation
could not be carried out by the \cpm algorithm: for as-skitter $k = 7$, WikiTalk $k = 8,9$,
Orkut $k = 6$ and Friendster $k = 5,6$.

The detail of all the calculated values, with which the following figures were generated, is available on
the following gitlab repository: \url{https://gitlab.lip6.fr/baudin/cpm-supplementary-material}.

\subsection{Comparison with the state of the art}
\label{sec:cpm-stateofart}

The algorithm proposed by Palla {\em et al.} in the original paper introducing
\cpm~\cite{palla2005uncovering} is quadratic in the number of maximal cliques. Given that we are
interested in graphs with at least several million cliques, these graphs are too big to be
processed by this algorithm. We do not perform experiments with this algorithm.

In addition, our tests have shown that the algorithm by Reid {\em et al.}~\cite{reid2012percolation}
has a better performance than that of Gregori {\em et al.}~\cite{gregori2013parallel}
(sequential version) therefore we do not present the results obtained with the version of
Gregori {\em et al}.

We observed that there are indeed linearity factors:
\begin{itemize}
\item in time: for each k-clique, each of its $(k-1)$-clique is processed in constant time, hence the
  running time of \cpm{} is indeed linear in $k \cdot n_k$
\item in memory: the memory is used to store the Union-Find structure on the $(k-1)$-cliques:
  one node per $(k-1)$-clique encoded on $k-1$ integers, hence the memory needed by \cpm{} is linear
  in $(k-1) \cdot n_ {k-1}$.
\end{itemize}

Figure~\ref{fig:cpm-time-mem} compares the time and the memory necessary for the computation of the
\cpm{} communities by our \cpm{} algorithm and the remaining competitive algorithms in the state of
the art, proposed by
Reid {\em et al.}~\cite{reid2012percolation} and Kumpula {\em et al.}~\cite{kumpula2008sequential}.
For each competitive algorithm, we plot its running time (resp. memory usage) divided by the running
time (resp. memory usage) of our \cpm{} algorithm.
We display the results as a function of $n_k$,
where $n_k$ is the number of $k$-cliques
of the input graph. 
In some cases, our \cpm{} algorithm obtains results whereas one of the state-of-the-art doesn't.
This can happen because this algorithm exceeds either the time or memory limit.
We display this by placing a symbol on the corresponding horizontal line at the top of the figure.

\begin{figure}
  \centering
  \includegraphics[scale=0.37]{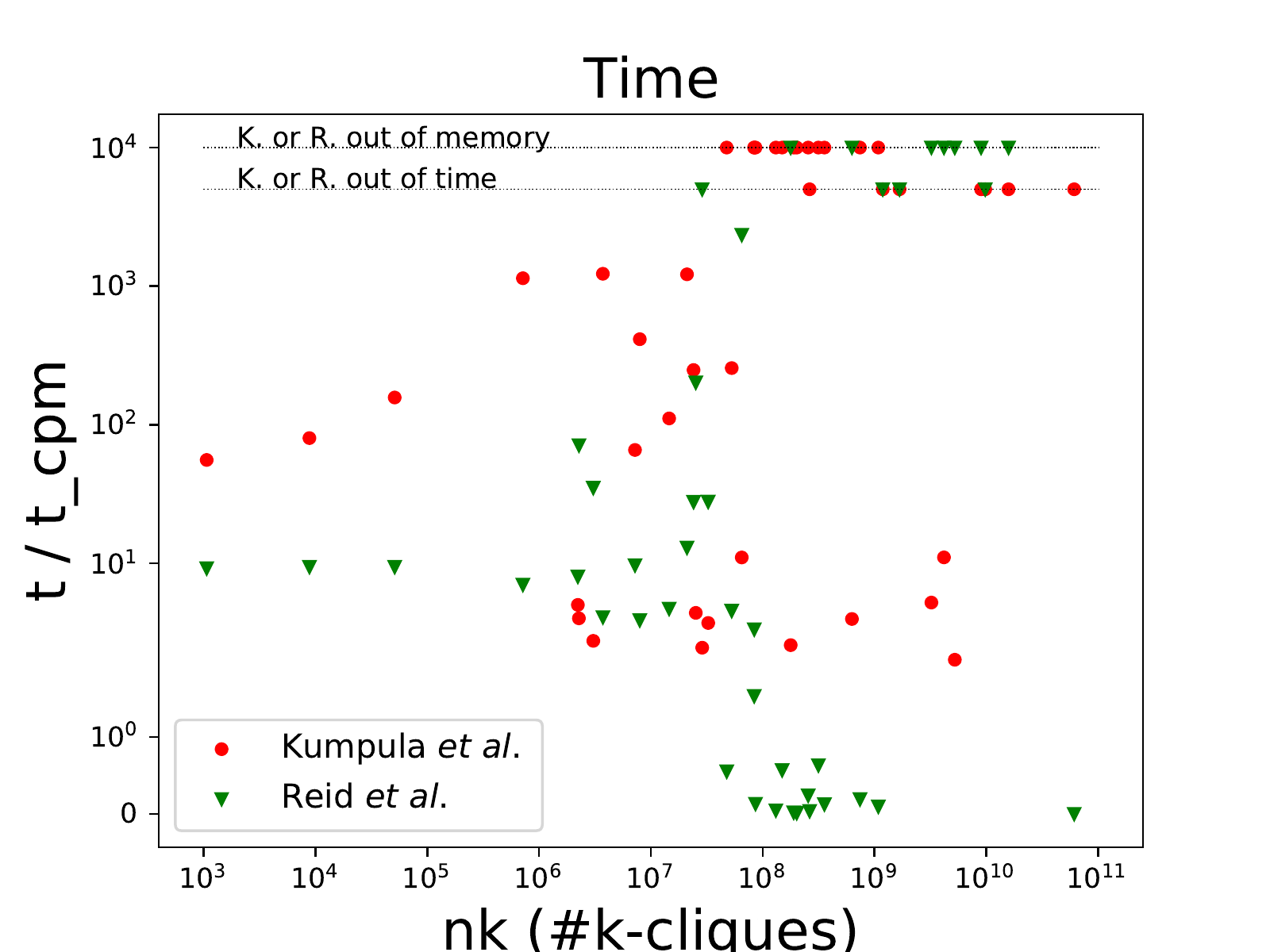}
  \hfill
  \includegraphics[scale=0.37]{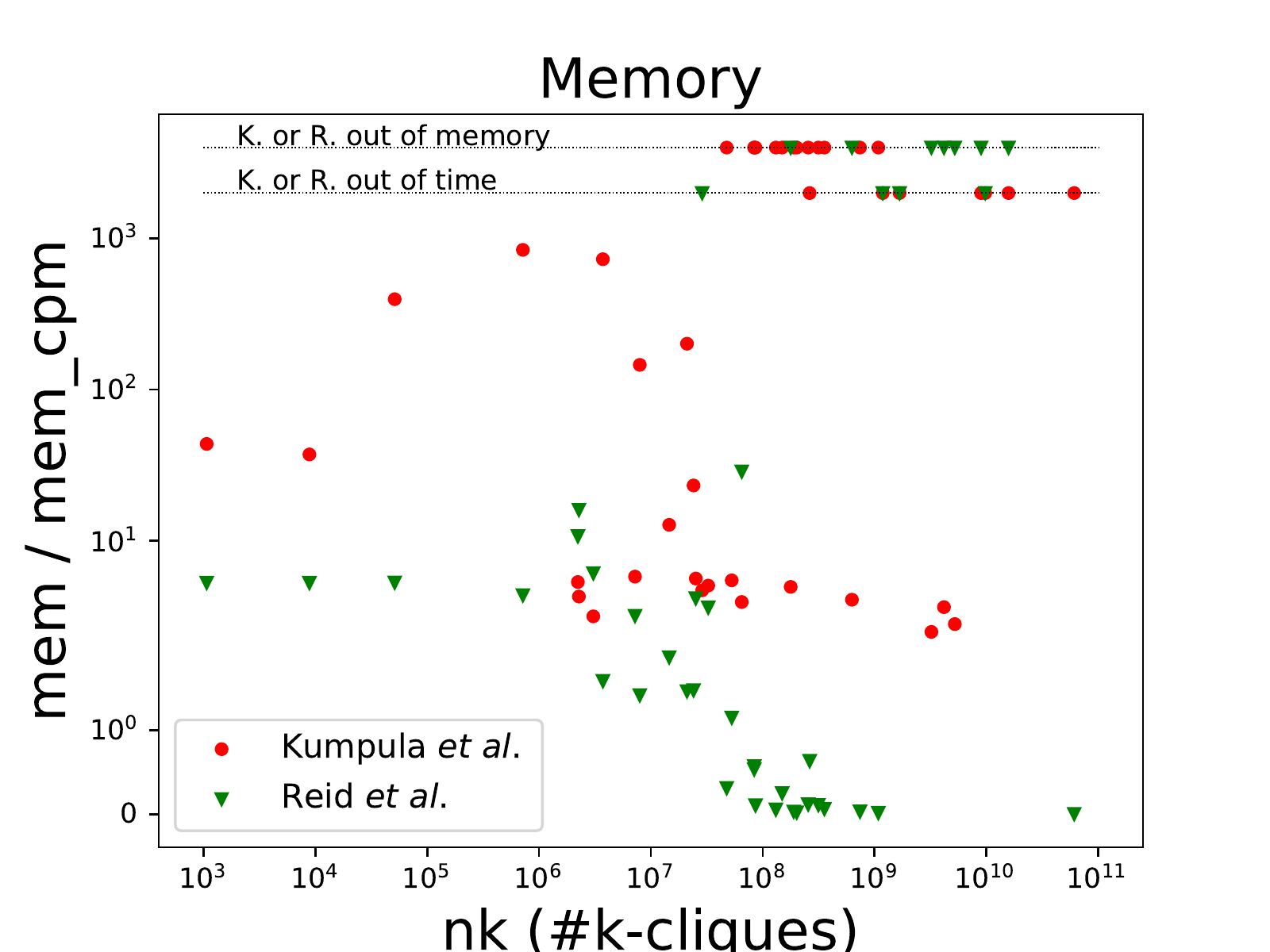} 
  \caption{Comparison between time and memory consumption of the state-of-the-art \cpm methods and
    those from ours. 
    For each competitive algorithm, we plot its running time (resp. memory usage) divided by the
    running time (resp. memory usage) of our \cpm{} algorithm.
    We display the results as a function of $n_k$,
    where $n_k$ is the number of $k$-cliques
    of the input graph.
    The maximum time is limited to 72h and the maximum memory to 390Gb.
    A marker placed at the corresponding line therefore indicates a computation that did not finish,
    because of either time or memory limit.
  }
  \label{fig:cpm-time-mem}
\end{figure}

First notice that the
Reid {\em et al.} algorithm is better than ours for the four smallest graphs
in our dataset (soc-pokec, loc-gowalla, Youtube, zhishi-baidu). Indeed, this algorithm begins by
computing the maximum cliques, then processes them to form communities. In the case of these
small graphs, the maximum cliques are easily computed.
For such graphs, the memory used does not depend on $k$ because in all cases
the maximal cliques are stored;
interestingly, the time computation time decreases with $k$ as only the maximal cliques of size
larger than or equal to $k$ have to be tested for adjacency.

This algorithm is also more efficient than ours in certain configurations, when the graph is already
well segmented into large cliques. This is the case with our DBLP graph, for which the Reid
{\em et al.} algorithm manages to compute the communities in 10 seconds when we need several hours
to process the large number of $k$-cliques.

Notice however that their algorithm does not allow to process the largest graphs of our
dataset. The as-skitter intermediate graph contains too many cliques and their algorithm does not
provide a result in less than 72 hours. For denser graphs (WikiTalk, Orkut, Friendster, LiveJournal),
there are too many maximum cliques for RAM, and the algorithm cannot run, while ours is able
to compute the result. 

Finally, the algorithm of Kumpula {\em et al.} is systematically less
efficient than ours.
Our algorithm is also able to obtain results in cases where no other algorithm can provide any
(see the points on the two horizontal lines on top of the figures).

\subsection{Memory gain of the \cpmz algorithm}

Figure~\ref{fig:cpmz-time-mem} (right) compares the  memory used by our algorithms \cpm and \cpmz with $z = 2,3$.
We show the memory used by \cpmz{} divided by the memory used by \cpm{}, as
a function of $n_k$.
As for the previous figure, we represent cases where \cpmz{} exceeds the time limit on a horizontal line on top of the figure.
In addition, cases where we obtain results with \cpmz{} and not \cpm{} are represented by symbols on a horizontal line at
the bottom of the figure.
For some small graphs, the number of $(k-1)$-cliques, which are stored by \cpm{}, is smaller than the number of $z$-cliques.
For these graphs \cpmz{} requires more memory than \cpm{}.
In most cases however,
we observe a huge memory gain, and in some cases it is even possible to obtain results
unachievable by our \cpm algorithm.

\begin{figure}
  \centering
  \includegraphics[scale=0.37]{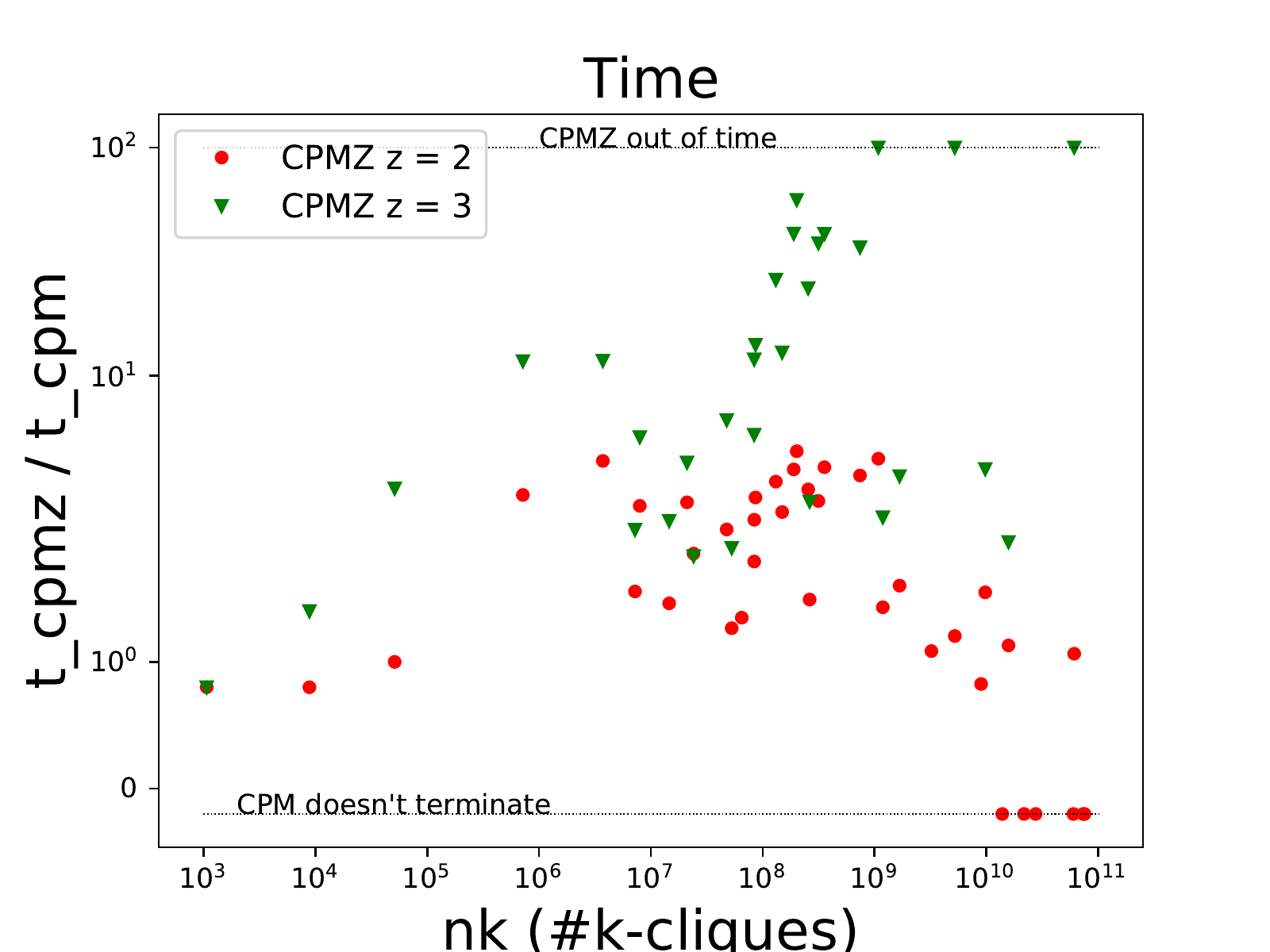}
  \includegraphics[scale=0.37]{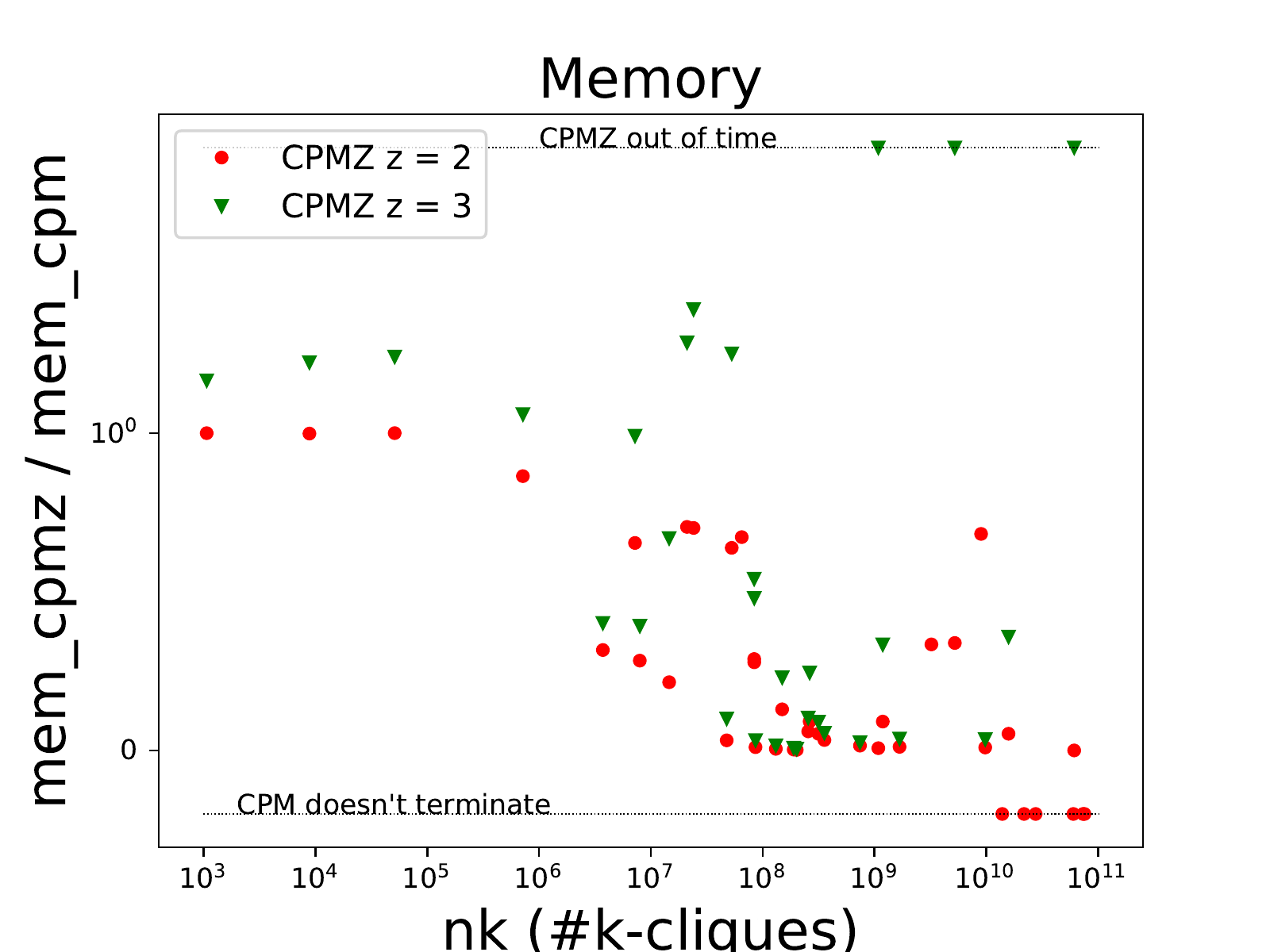}
  \caption{Comparison between time and memory consumption of the \cpmz algorithm and our \cpm.
    For \cpmz with $z=2$ and $z=3$, we plot its running time (resp. memory usage) divided by the running time
    (resp. memory usage) of our \cpm{} algorithm.
    We display the results as a function of $n_k$,
    where $n_k$ is the number of $k$-cliques
    of the input graph.
    The maximum time is limited to 72h, and the memory is not limiting in comparison with \cpm.
    A marker placed at the top line therefore indicates a computation that did not finish, because of time limit.
    A marker placed at the bottom line indicates a computation that finishes for \cpmz but not for \cpm.}
  \label{fig:cpmz-time-mem}
\end{figure}

In addition to this display and to the discussion about
memory requirements in Section~\ref{sec:analysis},
we performed experiments to evaluate the memory associated with each $z$-clique
in the algorithm. To do so, we computed the mean number of Union-Find nodes to which each $z$-clique belongs.
In Algorithm~\ref{algo:cpmz}, we therefore computed the sum the size of all \texttt{Setz[$c_z$]}
every time the algorithm reaches Line~\ref{cpmz:reduce-setz}
Then, we divided this sum by the number of times this line is browsed, which is $n_{k-1} \cdot k \cdot \binom{k-1}{z}$.
We observed that this factor remains low;
for all graphs and all values of $k$ and $z$ it is between 1 and 2, except for the small graphs
of Youtube with a value around 4 for $z=2$ and $k\in [10,15]$, $z=3$ and $k\in [11, 15]$,
and zhishi-baidu with a value around 5 for $z=2$ and \(k=5,6,7\).

Concerning the running time, as discussed in Section~\ref{sec:analysis},
there are two factors that cause \cpmz{} to be slower than \cpm{}.
The first one is the fact that a $z$-clique can belong to several Union-Find nodes.
As we just shown, this number is small in practice and therefore it does not play
a strong role in the running time.
The other factor is the extra $\binom{k-1}{z}$ factor induced by the fact that we
consider all $z$-cliques included in a $k$-clique.
This factor is high and therefore
the computation time remains the limiting factor.
Figure~\ref{fig:cpmz-time-mem} (left) compares it to the running time of \cpm{}.

\subsection{\cpmz communities are very close to \cpm{} communities}
\label{sec:cpmz-precision}

To measure the precision of our algorithm, we compare the agglomerated communities computed
by the \cpmz algorithm with those of the \cpm algorithm. To do so, we use an implementation of a Normalized
Mutual Information measure for sets of overlapping clusters, called {\sc onmi}, provided by McAid
{\em et al.}~\cite{mcdaid2013normalized}.
This tool measures how similar two sets of overlapping communities are.

We carried out the similarity comparisons on all the graphs of our dataset, with all the values of
$k$ for which we can compute the communities with the \cpm algorithm (see Table~\ref{tab:graphs}). We observe
the following:

\begin{itemize}[topsep=0ex,partopsep=0ex]
\item for \cpmz with $z = 2$, the average similarity is $98.6\%$, the median is $99.4\%$ and
  all values are larger than $93.8\%$;
\item for \cpmz with $z = 3$, the average similarity is $99.95\%$, the median is $100\%$ and all values are larger
  than $99.5\%$.
\end{itemize}

This confirms that the incorrect merges between communities performed by the \cpmz algorithm 
have little influence on the final result: the structure of communities is barely
impacted by the \cpmz algorithm. 


\section{Conclusion and discussions}
\label{conc}

In this paper we addressed the problem of overlapping community
detection on graphs through the clique percolation method (\cpm{}).
Our contributions are twofold:
first 
we proposed an improvement in the computation of the exact
result by leveraging a state of the art $k$-clique listing method;
then we proposed a heuristic algorithm called \cpmz{} that provides agglomerated
communities, i.e. communities that are supersets of the exact
communities; this algorithm uses much less memory than the exact
algorithm, at the cost of a higher running time.

Through extensive experimentations on a large set of graphs coming
from different contexts, we show that:

\begin{itemize}[topsep=0ex,partopsep=0ex]
\item our exact \cpm{} algorithm outperforms the state of the
  art algorithms in many cases,
  and we are able to compute the \cpm{} communities in cases
  where it was not possible before;
\item our relaxed \cpmz{} algorithm uses significantly less memory than the
  exact algorithm; even though its running time is higher, this allows us
  to obtain agglomerated communities in cases where no other algorithm
  can run;
\item finally the results provided by the \cpmz{} algorithm have an excellent
  accuracy, according to the {\sc onmi} method and
   obtain a score very close to 1 in the vast majority of cases.
\end{itemize}

Notice however that for the DBLP graph, which is of medium size, the approach proposed
in~\cite{reid2012percolation} works better than ours.
This can be explained because this graph naturally has a strong clique structure.
Indeed, a link exists in this graph if two authors have written a paper together,
and each paper therefore induces a clique on the set of its authors.
In this case computing the maximal cliques and extracting the community out of them
is more efficient than detecting adjacent $k$-cliques.
This raises the interesting question of whether it is possible to predict
which method will be more efficient on a graph by studying this graph's structure.

Several other interesting perspectives arise from our work.
It should be noted that the order in which $k$-cliques are
processed plays an important role
in the incorrect community agglomerations performed by the \cpmz{} algorithm,
that we do not 
yet fully understand.
Our experiments show that in practice only a few merges of
$k$-clique communities happen.
This gives rise to many interesting
graph theoretical questions about the
characterisation of the sub-graphs that can produce an incorrect $k$-clique adjacency detection:
how many of them are there in a typical real-world graph?
We conjecture that it is possible to construct examples in which no
processing order of the $k$-cliques will lead to the exact solution.
However, in many cases including real-world graphs, it is possible
that a certain processing order of $k$-cliques yields results of a
higher quality than other orderings.  This raises the question of how
to design such an ordering.
Another interesting possibility would be
to run the \cpmz{} algorithm with two or more different $k$-cliques
ordering and compare their output: since the \cpmz{} communities are
coarser than the exact \cpm{} communities, it is possible to compare
the communities of both outputs to obtain a better result that
any of the individual runs.

Finally, the linkstream formalism~\cite{latapy:hal-01665084} allows to represent interactions that
occur at different times,
which is a natural framework to represent people meeting at different time during
the week or computers exchanging {\sc ip} packets on the internet.
It would be interesting to investigate the community structure and its temporal aspects in such data
by extending the definition of $k$-clique communities to linkstreams.


\section*{Acknowledgements}

This work was partly supported by projects ANER ARTICO
(Bourgogne-Franche-Comté region), ANR (French National Agency of
Research) Limass project (under grant ANR-19-CE23-0010), ANR FiT
LabCom and ANR COREGRAPHIE project (grant ANR-20-CE23-0002).
We would like to greatly thank Lionel Tabourier for insightful
discussions, useful comments and suggestions, and Fabrice Lecuyer
for his careful proofreading.

\bibliographystyle{splncs04}
\bibliography{biblio} 

\end{document}